\documentclass[12pt]{amsart}
\usepackage{amsmath,amscd,amsfonts,amssymb,color}
\usepackage{latexsym, graphicx,enumerate}
\usepackage[pdftex,bookmarks,colorlinks]{hyperref}
\definecolor{darkblue}{rgb}{0.0,0.0,0.3}
\hypersetup{colorlinks=false,citecolor=darkblue,
unicode=true,pdftitle=entangled.pdf,pdfauthor=Ritabrata Sengupta,bookmarks=false}

\newtheorem{defin}{Definition}[section]
\newtheorem{rem}{Remark}[section]

\newtheorem{theorem}{Theorem}[section]

\newtheorem{corollary}{Corollary}[section]
\newtheorem{innercustomthm}{Theorem}
\newenvironment{customthm}[1]
{\renewcommand\theinnercustomthm{#1}\innercustomthm}
{\endinnercustomthm}
\newcommand{\tr}{\mathrm{Tr}}

\newcommand{\bra}[1]{\ensuremath{\left\langle #1\right|}}
\newcommand{\ket}[1]{\ensuremath{\left|#1\right\rangle}}

\begin{document}
\title{Entanglement properties of
positive operators with ranges in completely entangled
subspaces}
\author{R Sengupta}
\address{Theoretical Statistics and Mathematics Unit, Indian  Statistical Institute, Delhi Centre, 7 S J
S Sansanwal Marg, New Delhi 110 016, India}
\email[R
Sengupta]{\href{mailto:rb@isid.ac.in}{rb@isid.ac.in},
\href{mailto:ritabrata.sengupta@gmail.com}{ritabrata.sengupta@gmail.com}}
\author{Arvind}
\address{Department of Physical Sciences, Indian Institute
of Science Education \& Research, Mohali; Sector 81,
Knowledge City, Mohali 140 306,
India}
\email[Arvind]{\href{mailto:arvind@iisermohali.ac.in}{arvind@iisermohali.ac.in}}
\author{Ajit Iqbal Singh}
\address{Theoretical Statistics and Mathematics Unit, Indian  Statistical Institute, Delhi Centre, 7 S J
S Sansanwal Marg, New Delhi 110 016, India}
\email[Ajit Iqbal
Singh]{\href{mailto:aisingh@isid.ac.in}{aisingh@isid.ac.in},
\href{mailto:ajitis@gmail.com}{ajitis@gmail.com}}
\begin{abstract}
We prove that the projection on a completely entangled
subspace $\mathcal{S}$ of maximum dimension in a
multipartite quantum system obtained by
Parthasarathy\cite{krp1} is not positive under partial
transpose. We next show that several positive operators with
range in $\mathcal{S}$ also have the same property. In this
process we construct an orthonormal basis of $\mathcal{S}$ and
provide a linking theorem to link the constructions of
completely entangled subspaces due to Parthasarthy, Bhat and
Johnston.
\end{abstract}
\thanks{The authors thank K. R. Parthasarathy for insightful
discussion, critical comments and useful suggestions  for
the paper. They also thank Nathaniel Johnston for his
careful reading of the manuscript, critical comments,
corrections and suggestions which improved the paper. Ajit
Iqbal Singh thanks Indian National Science Academy (INSA)
for INSA Honorary Scientist position and Indian Statistical
Institute for invitation to visit the institute under the
scheme together with excellent research facilities and
stimulating atmosphere.}
\maketitle
\section{Introduction}
Entanglement is one of the key distinguishing features of
quantum mechanics which separates the quantum description of
the world from its classical counterpart.  Ever since its
discovery by Schr\"odinger~\cite{PSP:1737068, PSP:2027212}
and its use by  Einstein, Podolsky and Rosen~\cite{epr}, the
study of entanglement has played a central role in the area
of quantum theory and a huge volume of literature is
available  in this context.  In recent years, with the
emergence of quantum information where quantum entanglement
gets intimately connected to the  computational advantage of
quantum computers and to the security of quantum
cryptographic protocols, its study has become even more
important.  A detailed discussion on these topics is
available in the standard textbook of Nielsen and
Chuang~\cite{NC}, and a lucid introduction by
Parthasarathy~\cite{krp6} as well as a rigorous information
theoretic account by Wilde~\cite{Wilde} are also very useful
resources.

Entangled quantum states are those for which it is not
possible to imagine the physical reality of a composite
quantum system as two separate entities, even when there is
no active interaction between the two subsystems. In general
linear combinations of entangled states need not be
entangled, however, there have been constructions of
subspaces where every state in the subspace is entangled.
The first such construction was through the unextendable
product basis(UPB) by Bennett et. al. \cite{B2}, and
further extended by DiVincenzo et. al. \cite{T2}.  More recently,
Parthasarathy~\cite{krp1}, Bhat~\cite{bhat} and 
Johnston~\cite{PhysRevA.87.064302} have, by their own
different methods,  constructed
completely entangled subspaces $\mathcal{S}$ of maximum 
possible dimension in the state space of multipartite quantum
systems of finite dimensions. In such a subspace every state
in the subspace is entangled.

In our work we focus on projection operators on such
completely entangled subspaces.  We give a linking theorem
which links the constructions of Parthasarathy, Bhat and
Johnston.  Parthasarathy \cite{krp1} gave an orthonormal
basis for $\mathcal{S}$ for the bipartite case of equal
dimensions. We develop a method for construction of an
orthonormal basis for the space $\mathcal{S}$ in the general
case. Further, we construct the (orthogonal) projection on the space
$\mathcal{S}$ and show that it is not positive under partial
transpose at any level $j$.  The proof utilizes the
orthonormal basis for $\mathcal{S}$ that we develop.
Finally, we show that a large class of positive operators with range in
$\mathcal{S}$ are not positive under partial transpose at
 level $j$. This extends a substantial part of Johnston's result for the
bipartite case to the multipartite case by an altogether
different method.

The material in this paper is organized as follows: We begin
Section~\ref{spaces} with the basics of quantum
entanglement. We then describe the constructions of
completely entangled subspaces by Parthasarthy, Bhat and
Johnston. Next we give a theorem linking these three
constructions. Then we give a construction procedure of an
orthonormal basis for theses spaces. 
In Section~\ref{projections} we discuss our main results
regarding the entanglement properties of projection operators
on completely entangled subspace as also of certain positive
operators with support in this space. Section~\ref{conc}
offers some concluding remarks.

\section{Completely entangled spaces}
\label{spaces}
We begin with some well known
concepts and results. 

\subsection{Entanglement}
\begin{defin}
A finite dimensional quantum system is described by a finite
dimensional complex Hilbert space $\mathcal{H}$.   A
Hermitian, positive semidefinite operator
$\rho\in\mathcal{L(H)}$, the algebra of linear
operators on $\mathcal{H}$ to itself, with unit trace 
 is said to be a \emph{state} of the system
$\mathcal{H}$. Rank 1 states are called  pure states. A pure
state can be written as an outer product $\rho=|\psi\rangle\langle \psi|$ where
$|\psi\rangle\in\mathcal{H}$ and $\langle\psi|\psi\rangle=1$. 
\end{defin}
\begin{defin}
A state
$\rho$ acting on a bipartite system $\mathcal{H}_1\otimes\mathcal{H}_2$ is
said to be \emph{separable} if it can be written as 
\begin{equation}
\rho= \sum_{j=1}^m p_j \rho_j^{(1)}\otimes \rho_j^{(2)},
\quad p_j>0, \quad \sum_{j=1}^m p_j=1,
\end{equation}
where $\rho_j^{(1)}$ and  $\rho_j^{(2)}$ are states in the
system $\mathcal{H}_1$ and $\mathcal{H}_2$ respectively. 
\end{defin}
\begin{defin}
\par A state is said to be entangled, if it is not separable
by the above definition. Entangled states can be pure or
mixed. For an entangled pure state
$\rho=\ket{\psi}\bra{\psi}$, $\ket{\psi}$ is called an
entangled (unit) vector and any non-zero multiple of
$\ket{\psi}$ is called an entangled vector.  
\end{defin}
\par If the state is pure and separable, then it can be written in the
form $|\psi\rangle =|\psi_1\rangle \otimes |\psi_2\rangle$, and
hence $\rho=|\psi_1 \rangle\langle \psi_1| \otimes |\psi_2
\rangle\langle \psi_2|$. 
If we take partial trace with
respect to any of the subsystems, say $\mathcal{H}_2$,
then we get a pure state $\tr_{\mathcal{H}_2}\rho=|\psi_1
\rangle\langle \psi_1|$ as the reduced density matrix. On
the other hand, for an entangled pure state we always get a
mixed state after a partial trace.
Hence, a pure state is separable if and only if the reduced
density matrices are of rank one. This method does not work
for mixed states. 

\par We also
 consider multi-partite quantum systems, where the state
space given by  $\mathcal{H}=
\mathcal{H}_1 \otimes \cdots \otimes \mathcal{H}_k$; or in
short, $\bigotimes_{j=1}^k \mathcal{H}_j$. A product vector in this
multipartite system space is written  as 
$\ket{x_1}\otimes \cdots \otimes \ket{x_k}$, with
$\ket{x_j}\in \mathcal{H}_j$ or as 
$\ket{x_1, \cdots,x_k}$ or in short as $\bigotimes_{j=1}^k\ket{x_j}$.
The state of the system $\mathcal{H}$ can be entangled or
separable.  {\bf An important open problem in the field is
to determine whether an arbitrary  state $\rho$ of an arbitrary
quantum system ${\mathcal H}$, is entangled or separable.}
For further details regarding entanglement we refer the survey
article written by Horodecki et. al.~ \cite{RevModPhys.81.865}.

\par For general states, a very important one way condition
to check entanglement is by using \emph{partial transpose} (PT). If a
quantum state becomes non-positive after PT then
it is called NPT and if it remains positive after partial
transpose it is called PPT. NPT states are definitely
entangled and separable states are definitely PPT while PPT
states can be entangled or separable. PPT entangled states
are also called bound entangled states and their
characterization into entangled and separable is a major
open issue in the field. Checking PPT condition is also
known as the `Peres test' because of the significant work by
Peres \cite{P1}.  As remarked by DiVincenzo et.
al.~\cite{T2}, in the case of multipartite systems, the PPT
condition can not be used directly. We can check the PPT
property under every possible bipartite partitioning of the
state. We discuss this process in some  detail because of
its use in our work.

\begin{defin}\label{def:pt}
Let, for $1 \leq j \leq k$, $\{\ket{p_j}:
p_j=0,1,\cdots,\dim(\mathcal{H}_j)-1\}$ be an orthonormal
basis in $\mathcal{H}_j$. Let  $\rho\in
\mathcal{L}(\mathcal{H}_1\otimes \cdots \otimes
\mathcal{H}_k)$ be an operator. Then $\rho$ can 
be expressed in the form  
\begin{equation} \label{eq:pt1}
\rho=\sum_{p_1,q_1=0}^{\dim \mathcal{H}_1 -1}\cdots
\sum_{p_k,q_k=0}^{\dim \mathcal{H}_k -1}
\rho_{p_1,\cdots,p_k; q_1,\cdots,q_k} \ket{p_1,
 \cdots, p_k} \bra{q_1, \cdots, q_k}.
\end{equation}
The partial transpose of $\rho$, with respect to the $j$th
system, is given by 
\begin{multline}\label{eq:pt2}
\rho^{PT_j}=\sum_{p_1,q_1=0}^{\dim \mathcal{H}_1 -1}\cdots
\sum_{p_k,q_k=0}^{\dim \mathcal{H}_k -1}
\rho_{p_1,\cdots,p_k;q_1,\cdots,q_k} \\ \ket{p_1,
\cdots, p_{j-1}, q_j, p_{j+1}, \cdots, p_k} \bra{q_1,
\cdots, q_{j-1}, p_j, q_{j+1}, \cdots, q_k}.
\end{multline}

If for a state $\rho$,
$\rho^{PT_j}$ is positive, then $\rho$ is said to be
positive under partial transpose at the $j$th level, in
short, PPT$_j$. If a state $\rho$ is not PPT${}_j$, then it is
said to be not positive under partial transpose at the $j$th level, in
short, NPT$_j$.

\end{defin}

\begin{rem}
\noindent
\begin{enumerate}[(i)]
\item It is a fact that the property PPT${}_j$ is independent
of the choice of orthonormal basis in $\mathcal{H}_j$. 
\item In case of any bipartite system $\rho$, it is said to
be PPT if it is PPT${}_1$ or PPT${}_2$ (in this case PPT${}_1$
implies PPT${}_2$  and vice versa).
\item Woronowicz~\cite{wor} showed that, a state in
$\mathbb{C}^2\otimes \mathbb{C}^2$, $\mathbb{C}^2\otimes
\mathbb{C}^3$ or $\mathbb{C}^3\otimes \mathbb{C}^2$ is
separable if and only if it is PPT. For higher dimensions, PPT
is necessary, but not sufficient for separability and there are
examples of entangled states which are PPT.  
First examples of such states were constructed by
Choi~\cite{choi4} for $3\otimes3$,  Woronowicz~\cite{wor} for
$2\otimes4$ and later by St{\o}rmer~\cite{sto2}
for $3\otimes 3$.
\end{enumerate}
\end{rem}

\begin{defin}
 For any proper subset $E$ of $\{1,2,\cdots,k\}$ and its
complement $E'$ in $\{1,\cdots,k\}$ let
$\mathcal{H}(E)=\bigotimes_{j\in E} \mathcal{H}_j$ and
$\mathcal{H}(E')=\bigotimes_{j\in E'} \mathcal{H}_j$. Then
$\mathcal{H}=\mathcal{H}(E)\otimes\mathcal{H}(E')$. Any such
decomposition is called  a bipartite cut. 
A state $\rho\in
\mathcal{H}$ is said to be positive under partial transpose,
in short,  PPT if it is PPT under any
bipartite cut. 
\end{defin}

\begin{rem} \label{rem:cut}
Obviously if $\rho$ is PPT then $\rho$
is PPT$_j$ for each $j$; all we need to do is is to take
$E=\{j\}$. In other words, if $\rho$ is
NPT$_j$ for some $j$, then it is NPT.
\end{rem}
\subsection{Unextendable product bases} 
One well studied way to construct PPT entangled states
was given by Bennett et. al.~\cite{B2} by using
unextendable product basis.
\begin{defin}\label{def:upb}
An incomplete set of product vectors
$\mathcal{B}$ in the Hilbert space
$\mathcal{H}=\bigotimes_{j=1}^k \mathcal{H}_j$ is 
called unextendable if the space $\langle
\mathcal{B}\rangle^\perp$ does not contain any product
vector. The vectors in the set $\mathcal{B}$ 
are usually taken as orthonormal and are called unextendable
product bases, abbreviated as UPB.
\end{defin}
To avoid trivialities, we assume $\dim \mathcal{H}_j =d_j \geq
2$. Let $D=d_1d_2\cdots d_k$. Bennett et. al.~\cite{B2} gave
three examples of UPB for bipartite and tripartite systems
namely, PYRAMID, TILES and SHIFT.  We state the key theorem of
Bennett et. al.~\cite{B2} which allows one to construct PPT
entangled states from UPB and which is relevant to this paper.
\begin{customthm}{A}\label{th:bbd}~\cite{B2}
If in the Hilbert space $\mathcal{H} = \bigotimes_{j=1}^k
\mathcal{H}_j$ of dimension
$D=d_1\cdots d_k$, as above, there is a mutually orthonormal set of
unextendable product basis : $\{|\psi_s\rangle: s=1,\cdots,
d\}$, then the state 
\begin{equation}
\rho=\frac{1}{D-d}\left(I_D-\sum_{s=1}^d
|\psi_s\rangle\langle\psi_s|\right),\label{eq:bbd}
\end{equation}
where $I_D$ is the identity operator on $\mathcal{H}$, is an entangled state
which is PPT. 
\end{customthm}

\noindent The proof depends on the orthogonality of the basis
vectors $\ket{\psi_s}$.

\par The above theory was further extended by DiVincenzo et. al.~\cite{T2} to
include generalizations of the earlier examples to multipartite systems and a
complete characterization of UPB in $\mathbb{C}^3 \otimes \mathbb{C}^3$. There
is a large volume of literature in this area. Recently, Johnston has given
explicit computation of four qubit UPB~\cite{2014arXiv1401.7920J}.

\subsection{Entangled subspaces} Let  $\mathcal{H}=
\mathcal{H}_1\otimes\cdots \otimes \mathcal{H}_k$, where for
$1 \leq j \leq k$, $\mathcal{H}_j=
\mathbb{C}^{d_j}$ for some $d_j<\infty$ as above.
Wallach~\cite{MR1947343} considered the question of the maximal possible
dimension of a subspace $\mathcal{S}$  of $\mathcal{H}$ where each nonzero
vector is an entangled state. He called such 
subspaces entangled subspaces, as
they do not contain any nonzero product vector.  He showed that

\begin{customthm}{B}\cite{MR1947343} The dimension of a subspace, where each
vector is entangled, is $\leq d_1\cdots d_k -(d_1+\cdots +d_k) +k-1$.
Furthermore, this upper bound is attained.  
\end{customthm}
\subsection{Parthasarathy's construction}\label{ssec:krp}
Parthasarathy \cite{krp1} gave an explicit
construction of such entangled subspaces where the maximal
dimension is attained.  We note that Parthasarathy
calls such subspaces completely entangled subspaces. Let
$\mathcal{H}=\mathcal{H}_1 \otimes \cdots \otimes
\mathcal{H}_k$ be as above. 
Let $\lambda\in\mathcal{C}$. For $1\leq j \leq k$, let 
\begin{equation}
v_{\lambda, j}=
\begin{pmatrix}
1\\
\lambda\\
\lambda^2\\
\vdots\\
\lambda^{d_j-1}
\end{pmatrix}\equiv \sum_{x=0}^{d_j-1} \lambda^x \ket{x};
\end{equation}
where $\{\ket{x}:x=0,1,\cdots, d_j-1\}$ is the standard
basis of $\mathcal{H}_j=\mathbb{C}^{d_j}$. Set 
\begin{equation}
\ket{v_\lambda}\equiv v_{\lambda,1} \otimes \cdots \otimes
v_{\lambda,k}=\bigotimes_{j=1}^kv_{\lambda,j}.
\end{equation}
\par Set $N=\sum_{j=1}^k (d_j -1)=  \sum_{j=1}^k d_j -k$.
Choose any $(N+1)$ distinct complex numbers
$\lambda_0,\lambda_1,\cdots, \lambda_N$ and denote the
linear span of $\{v_{\lambda_n}: 0\leq n \leq N\}$ by
$\mathcal{F}$, i.e. $\mathcal{F}=\langle v_{\lambda_n}:
0\leq n \leq N \rangle $. Then $\{ v_{\lambda_n}:0\leq n
\leq N\}$ is a basis of $\mathcal{F}$. Consider the subspace
$\mathcal{S}=\mathcal{F}^\perp$.

\par It has been shown in~\cite{krp1}
that the space $\mathcal{S}$ does not contain any product vector and
is of dimension $M=d_1\cdots d_k-(d_1+\cdots+d_k)+k-1$.

\par Simple computations show that the basis vectors of
$\mathcal{F}$ need not all be orthogonal, but certain subspaces of
$\mathcal{F}$ can contain orthonormal basis of product
vectors. 

\par Another strong point in this paper is an explicit
construction of an orthonormal basis for $\mathcal{S}$ in
the case $k=2,~d_1=d_2$. We shall come back to this later
in \S \ref{ssuc:basiskrp} below.

\subsection{Bhat's construction~\cite{bhat}}\label{subsc:bhat}
For notational convenience, he starts with an infinite
dimensional space with an orthonormal basis  $\{e_0,e_1,\cdots\}$ and
identifies
$\mathcal{H}_r=\langle\{e_0,\cdots,e_{d_r -1}\}\rangle,~1\leq
r \leq k$,  and sets 
$\mathcal{H}=\mathcal{H}_1\otimes \cdots \otimes
\mathcal{H}_k$. 

Let $N=\sum_{r=1}^k (d_r-1)$. For $0\leq n \leq N$, let
$\mathcal{I}_n=\{\mathbf{i}=(i_r)_{r=1}^k,~ 0\leq i_r
\leq d_r-1 \text{ for } 1\leq r \leq k,~\sum_{r=1}^k i_r=n\}$. 
Let
$\mathcal{I}=\bigcup_{n=0}^N\mathcal{I}_n$. For $\mathbf{i} \in
\mathcal{I}$, let $e_{\mathbf{i}}=\bigotimes_{r=1}^k e_{i_r}$. 
For
$0\leq n\leq N$, let $\mathcal{H}^{(n)}=\langle \{ e_{\mathbf{i}}:
\mathbf{i}\in\mathcal{I}_n\}\rangle$. Then  $\{e_{\mathbf{i}}:\mathbf{i}\in
\mathcal{I}_n\}$ is an orthonormal basis for
$\mathcal{H}^{(n)}$. Further, $\mathcal{H}=\bigoplus_{n=0}^N
\mathcal{H}^{(n)}$ and    $\{e_{\mathbf{i}}:\mathbf{i}\in
\mathcal{I}\}$ is an orthonormal basis for $\mathcal{H}$.
Let  $0\leq n\leq N$. Let
$u_n=\sum_{\mathbf{i}\in\mathcal{I}_n} e_{\mathbf{i}}$. 
Let $\mathcal{T}^{(n)}=\mathbb{C}u_n$, then
$\mathcal{H}^{(n)}=\mathcal{S}^{(n)}\bigoplus T^{(n)}$, where 
\[\mathcal{S}^{(n)}=\mathrm{span}\{e_{\mathbf{i}}-e_{\mathbf{j}}:\mathbf{i},\mathbf{j}\in
\mathcal{I}_n\}. \]
Clearly $\mathcal{S}^{(n)}$ is also equal to the set of all the sums
$\sum_{\mathbf{i}\in\mathcal{I}_n}\alpha_{\mathbf{i}}e_{\mathbf{i}}$
such that $\sum_{\mathbf{i}\in\mathcal{I}_n}\alpha_{\mathbf{i}}=0$.
Further, $\mathcal{S}^{(0)}=\{0\}=\mathcal{S}^{(N)}$. Let
 $\mathcal{T}=\bigoplus_{n=0}^N\mathcal{T}^{(n)}$ and
$\mathcal{S}_B=\bigoplus_{n=0}^N \mathcal{S}^{(n)}$, which
is the same as $\bigoplus_{n=1}^{N-1} \mathcal{S}^{(n)}$. Then
$\mathcal{S}_B^\perp=\mathcal{T}$ and
$\mathcal{H}=\mathcal{S}_B\oplus \mathcal{T}$.
\begin{customthm}{C}\cite{bhat}
$\mathcal{S}_B$ is a completely entangled subspace of maximal
dimension.
\end{customthm}

\begin{rem}\label{rem:bhat}
\noindent
\begin{enumerate}[(i)]
\item We note that for $\lambda\in\mathbb{C}$,
\begin{eqnarray}
\ket{z^\lambda} & \equiv & \left(\sum_{j_1=0}^{d_1-1}
\lambda^{j_1}
e_{j_1}\right) \otimes \cdots \otimes
\left(\sum_{j_k=0}^{d_k-1} \lambda^{j_k}
e_{j_k}\right)\nonumber\\
& = & \sum_{n=0}^N \lambda^n
\left(\sum_{\mathbf{i}\in\mathcal{I}_n}
e_{\mathbf{i}}\right)\\
&=& \sum_{n=0}^N \lambda^n u_n.\nonumber
\end{eqnarray} 
 
\item We now consider  $\mathcal{H}_r$'s as subspaces of
$\mathbb{C}^\delta$, with $\delta=\max_{j=1}^k d_j$ and $e_s
\equiv \ket{s}$ for $1\leq s \leq \delta$. So we can identify
$\ket{v_\lambda}$ and $\ket{z^\lambda}$.  Let $\lambda_n,~0\leq n \leq
N$ be distinct complex numbers as in \S \ref{ssec:krp}. 
Then $\{\ket{v_{\lambda_n}}: 0 \leq n \leq N\}$
is a linearly independent subset of $\mathcal{T}$. So
$\mathcal{F}=\mathcal{T}$. This also shows that $\mathcal
{F}$ is independent of the choice of complex numbers. Thus 
\[\mathcal{S}=\mathcal{F}^\perp=\mathcal{T}^\perp=\mathcal{S}_B.\]
\end{enumerate}
\end{rem}

\begin{customthm}{D}\cite{bhat}
The set of product vectors in $\mathcal{S}^\perp=\mathcal{T}$ is 
\[\{c\ket{z^\lambda}: c\in \mathbb{C}, \lambda \in\mathbb{C} \cup
\{\infty\}\};\]
where $\ket{z^\infty} = \bigotimes_{r=1}^k e_{d_r-1}$.
 \end{customthm}

\subsection{Johnston's construction \cite{PhysRevA.87.064302}}
Johnston concentrated on constructing a completely entangled subspace
$\mathcal{S}_J$ of $\mathbb{C}^{d_1}\otimes\mathbb{C}^{d_2}$
of dimension $(d_1-1)(d_2-1)$ for bipartite systems such that 
every density matrix with range contained in it is NPT. In the 
notation  of Subsections~\ref{ssec:krp} and~\ref{subsc:bhat}, 
\begin{equation}
\mathcal{S}_J =\left\langle \{ w_{x,y}=\ket{x}\otimes \ket{y+1} -
\ket{x+1}\otimes  \ket{y}: 0 \leq x \leq d_1-2, 0\leq y \leq d_2 -2
\} \right \rangle.
\end{equation}
We end this subsection with our theorem which establishes an
interesting and useful link between different constructions of 
completely entangled subspaces.
\begin{theorem}
For the bipartite case, the completely entangled spaces 
$\mathcal{S},~\mathcal{S}_B$ and $\mathcal{S}_J$ can be 
identified with each other. 
\end{theorem}
\begin{proof}
In view of Remark~\ref{rem:bhat} and the discussion in this
section, we only need to note that for $0\leq
x \leq d_1 -2$ and $0\leq y \leq d_2-2$, $w_{x,y} \in
\mathcal{S}^{(x+y+1)}$. Thus
$\mathcal{S}_J\subseteq\mathcal{S}_B$. But $\dim
\mathcal{S}_B =(d_1-1)(d_2-1)=\dim\mathcal{S}_J$. Hence
$\mathcal{S}_B=\mathcal{S}_J$. 
\end{proof}

\subsection{Parthasarathy's orthonormal basis for $\mathcal{S}$ for
bipartite case of equal dimensions~\cite{krp1}}\label{ssuc:basiskrp}

\par We need the following explicit construction of the
orthonormal 
basis $\mathcal{B}$ of $\mathcal{S}$ given in \cite{krp1}
for the bipartite case $\mathcal{H}=\mathcal{H}_1 \otimes
\mathcal{H}_2$, with $d_1 = d_2 =\nu$, say.
\begin{enumerate}[(a)]
\item Antisymmetric vectors: 
\begin{equation*}\label{antsym}
\ket{a_{x,y}}=\frac{1}{\sqrt{2}}(|xy\rangle
-|yx\rangle),~~0\leq x<y \leq \nu-1.
\end{equation*}
\item For $2\leq n \leq \nu-1$ and $n$ even, vectors of the
forms :
\begin{gather*}
\ket{b_0^n}= \frac{1}{\sqrt{n(n+1)}} \left( \sum_{m=0}^{\frac{n}{2}-1}
(|m, n-m\rangle + |n-m, m\rangle) -n
\left|\frac{n}{2},\frac{n}{2}\right\rangle \right),
\label{sym1}
\quad \text{and}\\ 
\ket{b_p^n}= \frac{1}{\sqrt{n}} \sum_{m=0}^{\frac{n}{2}-1}
\exp\left(\frac{4\pi\imath mp}{n}\right)
(|m, n-m\rangle + |n-m, m\rangle), \quad 1\leq p \leq
\frac{n}{2}-1. \label{sym2} 
\end{gather*}
\item  For $2\leq n \leq \nu -1$ and $n$ odd, vectors of
the form:
\begin{equation*}\label{sym3}
\ket{b_p^n}= \frac{1}{\sqrt{n+1}} \sum_{m=0}^{\frac{n-1}{2}}
\exp\left(\frac{4\pi\imath mp}{n+1}\right)
(|m, n-m\rangle + |n-m, m\rangle), \quad 1\leq p \leq
\frac{n-1}{2}. \end{equation*}
\item  For $\nu\leq n \leq 2\nu-4$ and $n$ even, vectors of
the form:
\begin{gather*}
\begin{split}
\ket{b_0^n}= \frac{1}{\sqrt{(2\nu -2-n)(2\nu-1-n)}} \left(
\sum_{m=0}^{\frac{2\nu-2-n}{2}-1}
(|n-\nu+m+1, \nu-m-1\rangle  \right.\\+ |\nu-m-1,
n-\nu+m+1\rangle) -(2\nu-2-n)
\left. \left|\frac{n}{2},\frac{n}{2}\right\rangle
\right),\quad \text{and} \end{split} \label{sym4}
\\ 
\begin{split}
\ket{b_p^n}= \frac{1}{\sqrt{2\nu-2-n}} \sum_{m=0}^{\frac{2\nu-2-n}{2}-1}
\exp\left(\frac{4\pi\imath mp}{2\nu-2-n}\right)
(|n-\nu+m+1, \nu-m-1\rangle \\ 
 + |\nu-m-1, n-\nu+m+1\rangle), \quad 1\leq p \leq
\frac{2\nu-2-n}{2}-1. \end{split}\label{sym5}
\end{gather*}
\item  For $\nu\leq n \leq 2\nu-4$ and $n$ odd, vectors of
the form:
\begin{equation*}\label{sym6}\begin{split}
\ket{b_p^n}= \frac{1}{\sqrt{2\nu-1-n}} \sum_{m=0}^{\frac{2\nu-1-n}{2}-1}
\exp\left(\frac{4\pi\imath mp}{2\nu-1-n}\right)
\left(|n-\nu+m+1, \nu-m-1\rangle \right.\\ 
\left. + |\nu-m-1, n-\nu+m+1\rangle\right	), \quad 1\leq p \leq
\frac{2\nu-1-n}{2}-1. \end{split}
\end{equation*}
\end{enumerate}

\begin{rem}\label{rem:pack}
\noindent 
\begin{enumerate}[(i)]
\item An interesting aspect of $\mathcal{B}$ is that  for $1
\leq n \leq 2 \nu -3$, $\mathcal{B}_n =\mathcal{B} \cap
\mathcal{S}^{(n)}$ is an orthonormal basis for $\mathcal{S}^{(n)}$.
\item $\mathcal{B}_1=\{\ket{a_{0,1}}\}$ and
$\mathcal{B}_{2\nu -3}=\{\ket{a_{\nu-2,\nu-1}}\}$.
\item For $1 \leq g \leq \nu -2 $, $\ket{f_g}=\ket{g} \otimes
\ket{g}$ occurs as  a summand of  exactly one vector  in
$\mathcal{B}$. Further, $\ket{f_{\nu-1}}= \ket{\nu-1}\otimes
\ket{\nu-1}$ does not occur as a summand of vectors in $\mathcal{B}$.
\item Let $\mathbb{F}:\mathcal{H}_1\otimes \mathcal{H}_2
\rightarrow \mathcal{H}_2\otimes \mathcal{H}_1$ be the
linear operator, called FLIP or SWAP, satisfying
$\mathbb{F}(\ket{\xi}\otimes \ket{\eta}) = \ket{\eta}\otimes
\ket{\xi}$ for $\ket{\xi}\in \mathcal{H}_1$ and
$\ket{\eta}\in \mathcal{H}_2$. Then
$\mathbb{F}(\ket{a_{x,y}})=-\ket{a_{x,y}}$,
whereas $\mathbb{F}(\ket{b_p^n})=\ket{b_p^n}$;
$\ket{a_{x,y}}$ and $\ket{b_p^n}$ are as above. 
\end{enumerate}
\end{rem}

\subsection{Bhat's orthonormal basis for $\mathcal{S}$}
Bhat~\cite{bhat} indicated how to construct an
orthonormal basis for $\mathcal{S}$. He has also given
expressions for dimensions of $\mathcal{H}^{(n)}$ viz.,
$|\mathcal{I}_n|$ for $ 1\leq n \leq N$. In fact,
$\mathcal{I}_n=$ the coefficient of $x^n$  in the
polynomial $p(x) =\prod_{r=1}^k (1+x+\cdots +x^{d_r-1}) =$
number of partitions of $n$ into $(i_1,\cdots, i_k)$ with $
0\leq i_r \leq d_r-1$ for $1\leq r \leq k$. For instance, for $k=2, ~  d_1
\leq d_2 $, 
\[|\mathcal{I}_n| =\left\{
\begin{array}{lll}
n+1 & \text{ for }& 0\leq n \leq d_1-1 \\
d_1 &\text{ for }&d_1-1< n \leq d_2 -1 \\
d_1+d_2 -(n+1) &\text{ for }& d_2 -1 < n \leq d_1 +d_2 -2.
\end{array}\right.
\]
If $d_i=2$ for all $i$, then $|\mathcal{I}_n|
=\binom{k}{n},~0\leq n \leq k$. 

\subsection{Two useful techniques}
\par We now display techniques to be used in constructing
an orthonormal basis for the general bipartite and
multipartite case suitable for our purpose. 

\begin{theorem} \label{th:2nd}
Let $Y$ be a $d$-dimensional Hilbert space with $2\leq d <
\infty$ and $\{\ket{y_s}: 0\leq s \leq d-1\}$ an orthonormal
basis for $Y$. Let $Z$ be the subspace $\{\sum_{s=0}^{d-1}
\alpha_s\ket{y_s}: \sum_{s=0}^{d-1}
\alpha_s=0\}$.
\begin{enumerate}[(i)]
\item If $d=2$ then $Z=\mathbb{C}(\ket{y_0}-\ket{y_1})$.
\item Let $d\geq 3$. Then there exists an orthonormal basis
$\{\ket{z_s}: 0\leq s \leq d-2\}$ for $Z$ such that $\ket{y_0}$
occurs as a summand in $\ket{z_0}$ and $\ket{z_1}$; further,
for $d>3$, $\ket{y_0}$ does not occur as a summand in
$\ket{z_s}$ for  $2\leq s \leq d-2$. 
\item Let $d\geq3$. Let $1\leq r \leq d-2$. Let 
\[Z_r^1
=\left\{\sum_{s=0}^r \alpha_s \ket{y_s}: \sum_{s=0}^r
\alpha_s=0\right\} = Z\cap \langle \{\ket{y_s} : 0\leq s \leq
r\}\rangle,\] 
and
\[Z_r^2 =\left\{\sum_{s=r+1}^{d-1} \alpha_s \ket{y_s}:
\sum_{s=r+1}^{d-1} \alpha_s =0\right\}.\]
Let $\mathcal{C}_r^1 = \{\ket{z_s}: 0\leq s \leq r-1 \}$
 be an orthonormal basis for $Z_r^1$
such that $\ket{y_0}$ occurs as a summand in $\ket{z_0}$ and
in no other $\ket{z_s}$ for $s\leq r-1$. Then there exists
an orthonormal basis $\{\ket{z_s}:0\leq s \leq d-2\}$ for
$Z$ such that $\ket{y_0}$ occurs as a summand in $\ket{z_0}$ and
 $\ket{z_r}$ and in no other $\ket{z_s}$ for $0\leq s \leq
d-2$. 
\end{enumerate} 
\end{theorem}

\begin{proof}
\begin{enumerate}[(i)]
\item is immediate.
\item Let $\ket{z_0}=\frac{1}{\sqrt{2}}
(\ket{y_0}-\ket{y_1}), ~ \ket{\eta} =  (\ket{y_0}+\ket{y_1})$
and $\ket{v}=\sum_{s=2}^{d-1} \ket{y_s}$. Set 
\[ \ket{z_1}=\frac{(d-2)\ket{\eta} -2 \ket{v} }{\sqrt{2d(d-2)}}.\]
Then $\ket{y_0}$ occurs as a summand in $\ket{z_0}$ and
$\ket{z_1}$. 

\par We now consider the case $d>3$ and follow the
notation in (iii). We choose any orthonormal basis for
$Z_1^2$. For instance, we may choose the Fourier basis
\[\ket{z_p}=\frac{1}{\sqrt{d-2}}\sum_{s=2}^{d-1} \exp\left[ \frac{2\pi\imath
(s-2) (p-1)}{d-2}\right]\ket{y_s},\quad 2\leq p \leq d-2.\]

\item Let $\ket{\eta}=\sum_{s=0}^r \ket{y_s},~\ket{v} =
\sum_{s=r+1}^{d-1} \ket{y_s}$. Consider any $\ket{\xi}
=\sum_{s=0}^{d-1} \alpha_s \ket{y_s}$. For $0\leq s' \neq
s'' \leq r$, $\ket{y_{s'}} - \ket{y_{s''}} \in Z_r^1$. So
$\ket{\xi}\perp Z_r^1$ only if $\alpha_{s'} = \alpha_{s''}$
for $s'\neq s''$ with  $0 \leq s' \neq s'' \leq r$. Thus any
such vector has the form
\begin{equation}\label{eq:61}
\ket{\xi} =\alpha \ket{\eta} + \sum_{s=r+1}^{d-1} \alpha_s
\ket{y_s} \quad \text{with} \quad (r+1)\alpha +
\sum_{s=r+1}^{d-1} \alpha_s =0.
\end{equation}
Also any $\ket{\xi}$ of the form as in (\ref{eq:61}) is orthogonal
to $Z_r^1$. Set 
\[\ket{z_r}= \frac{(d-1-r) \ket{\eta} -(r+1)
\ket{\nu}}{\sqrt{d(r+1)(d-r-1)}}.\]
Then $\ket{y_0}$ occurs as a summand in $\ket{z_r}$. 

\par We now consider the case $r\leq d-3$, which forces 
$d\geq4$ for sure. Now
$\ket{\xi}$ as in (\ref{eq:61}), satisfies
$\langle\xi|z_r\rangle =0$ if and only if $\alpha=0$ if and
only if $\sum_{s=r+1}^{d-1} \alpha_s =0$ if and only if $\ket{\xi}$
has the form 
\[\ket{\xi} =\sum_{s=r+1}^{d-1} \alpha_s \ket{y_s}, \quad
\sum_{s=r+1}^{d-1} \alpha_s =0 \quad \text{ if and only if }
\quad \ket{\xi} \in Z_r^2.\]
As in the proof of (ii),  we choose any orthonormal basis
for $Z_r^2$. for instance, we may choose the Fourier basis,
\[ \ket{z_p} = \frac{1}{\sqrt{d-1-r}}\sum_{s=r+1}^{d-1} \exp\left[\frac{2\pi
(s-r-1) (p-r) } {d-1-r} \right] \ket{y_s},\quad r+1 \leq p
\leq d-2. \]
Then $\ket{y_0}$ does not occur as a summand in $\ket{z_p},~
r+1 \leq p \leq d-2$.

\end{enumerate}
\end{proof}

\subsection{Orthonormal basis for $\mathcal{S}$ (general
case)} \label{ssu:or}

\par We shall now construct a suitable orthonormal basis for
$\mathcal{S}$  in our multipartite system
$\mathcal{H}=\bigotimes_{j=1}^k \mathcal{H}_j$. Let $1\leq j
\neq j' \leq k$.  Set $\nu =\min \{d_j, d_{j'}\}$ and $\nu'
=\max \{d_j, d_{j'}\}$. We concentrate on the case $(k-2) +
(\nu' -\nu)>0$, as the remaining case $k=2,~ \nu=\nu'$ comes
under \S\ref{ssuc:basiskrp} above. It is enough to
construct suitable orthonormal basis for
$\mathcal{S}^{(n)}$ for  $1 \leq n \leq N-1$, because we can
just put them together to get an orthonormal basis for
$\mathcal{S}$. Let  $1 \leq n \leq N-1$. We take
$X=\mathcal{H}^{(n)}, ~ Z= \mathcal{S}^{(n)}$ in the
above theorem. We note that $\mathcal{H}^{(n)}$ has
dimension $d=|\mathcal{I}_n|$. For $0 \leq x, x' \leq \nu -1$, we take
$\mathbf{i}^{(x,x')}\in \mathcal{I}$ given by 
\[ i_t^{(x,x')} = \left\{ 
\begin{array}{ll}
0 & t \neq j \text{ or } j'\\
x & t=j\\
x' & t=j'.
\end{array}
\right.\]
At times we shall replace $\mathbf{i}^{(x,x')}$ by
$\widetilde{(x,x')}$. For $\ket{\xi} \in \mathbb{C}^\nu
\otimes \mathbb{C}^\nu$, we take $\tilde{\ket{\xi}}$ to
be the vector in $\mathcal{H}$ which is obtained by
considering $\ket{\xi}$ as a member of $\mathcal{H}_j\otimes
\mathcal{H}_{j'}$ and then filling in the remaining places
by $\ket{0}$ (if any). Then $\tilde{\mathcal{B}_n} = \{
\tilde{\ket{\xi}}: \ket{\xi} \in \mathcal{B}_n\}$ may be
thought of as an orthonormal basis for its linear span which is a
part of $\mathcal{S}^{(n)}$. 
\par  Let 
\begin{eqnarray*}
\mathcal{I}_n^1 &=& \left\{
\begin{array}{ll}
 \{\mathbf{i}\in \mathcal{I}_n, ~  0 \leq  i_j , i_{j'} \leq
\nu - 1, \text{ and } i_t = 0  \text { for } t \neq j,j'\}, & 1 \leq n \leq 2\nu -3\\
&\\
\emptyset & \text{otherwise.}
\end{array}\right. \\
 \mathcal{I}_n^2 &=& \mathcal{I}_n \setminus \mathcal{I}_n^1.
\end{eqnarray*}
We note that $|\mathcal{I}_n^1|$ is either $0$ or $\geq 2$. For
$n= 2g$ with $1\leq g \leq \nu -1$, we take $\mathbf{i}^0
=\widetilde{(g,g)}$. For $n=2g-1,~1\leq g \leq \nu -1$ we
take $\mathbf{i}^0 =(g-1,g)$. Next, for $1\leq n \leq 2\nu
-3$, we arrange members of $\mathcal{I}_n^1\setminus
\{\mathbf{i}^0\}$ in any sequence, say $\mathbf{i}^1,\cdots
, \mathbf{i}^{|\mathcal{I}_n^1|-1}$ insisting, for
$n=2g-1,~\mathbf{i}^1 =(g , g-1)$. 
Then, we arrange members of $\mathcal{I}_n^2$, if
any, in any manner we like. 
This will complete the enumeration of $\mathcal{I}_n$ as
$0,1,\cdots, |\mathcal{I}_n|-1$. For $n=2\nu -2$, we
enumerate $\mathcal{I}_n\setminus \{\mathbf{i}^0\}$ as   $\mathbf{i}^1,\cdots
, \mathbf{i}^{|\mathcal{I}_n|-1}$. For $2\nu -1 \leq n \leq
N-1$, we enumerate $\mathcal{I}_n$ in any manner we like as
$\mathbf{i}^0, \mathbf{i}^1,\cdots,
\mathbf{i}^{|\mathcal{I}_n|-1}$. Finally, we set
$\ket{y_s}=\ket{\mathbf{i}^s},~0\leq s \leq d-1
=|\mathcal{I}_n|-1$ and, in case $1 \leq n \leq 2\nu -3$,
 $r=|\mathcal{I}_n^1|-1$.

\par  To distinguish constructions for different $n$'s, we may use
extra fixture $n$; for instance  $^n\mathbf{i}^0,~
^n\mathbf{i}^1, \cdots, \ket{\eta_n},\ket{v_n}$ etc. in
place of  $\mathbf{i}^0, \mathbf{i}^1, \cdots,
\ket{\eta},\ket{v}$. 

\par This discussion combined with Theorem~\ref{th:2nd} above
immediately gives us the following theorem. 

\begin{theorem}\label{th:3rd}
Let $\mathcal{H}=\bigotimes_{t=1}^k \mathcal{H}_t$. Let
$1\leq j \neq j' \leq k, ~ \nu =\min\{ d_j , d_{j'} \} \leq
\nu' = \max \{d_j, d_{j'} \}$ and $(k-2) + (\nu' - \nu) >0$. There exists an orthonormal basis $\mathcal{C}$
for $\mathcal{S}$ such that  
\begin{enumerate}[(i)]
\item $\widetilde{\ket{0}\otimes \ket{0}}$ does not occur as a summand
in any vector in $\mathcal{C}$.

\item For $1 \leq g \leq \nu -2 , ~\widetilde{\ket{g}
\otimes \ket{g}}$ occurs as a summand in two members of
$\mathcal{C}$.

\item $(\widetilde{\ket{\nu-1}\otimes \ket{\nu-1}})$ occurs
as a summand in two members of $\mathcal{C}$ except for the
bipartite case with $2=\nu  < \nu'$ or $\nu'= \nu +1$, when it occurs only once. 

\item For $2 \leq g \leq \nu -1,~\widetilde{(\ket{g-1}\otimes
\ket{g})}$ and $\widetilde{(\ket{g}\otimes
\ket{g-1})}$  occur as a summand in (the same) two members
of $\mathcal{C}$.

\item In particular,
$\widetilde{(\ket{0}\otimes\ket{1})}$,
$\widetilde{(\ket{1}\otimes\ket{0})}$ and $\widetilde{(\ket{1}\otimes\ket{1})}$, occur as summands as
follows.

\begin{enumerate}
\item Vectors $\widetilde{\ket{0}\otimes\ket{1}}$ and
$\widetilde{\ket{1}\otimes\ket{0}}$ occur as a summand in
$\widetilde{\ket{a_{0,1}}} = \frac{1}{\sqrt{2}} (
\widetilde{\ket{0}\otimes\ket{1}} -
\widetilde{\ket{1}\otimes\ket{0}})$, and in case $k\geq 3$,
also in $\ket{c_0^1}
=\frac{1}{\sqrt{2|\mathcal{I}_1| (|\mathcal{I}_1|-2)}}
\left(  (|\mathcal{I}_1|-2)
(\widetilde{\ket{0}\otimes\ket{1}} +
\widetilde{\ket{1}\otimes\ket{0}}) -2 \ket{v_1}\right) =
\frac{1}{\sqrt{2k (k-2)}}
\left(  (k-2)
(\widetilde{\ket{0}\otimes\ket{1}} +
\widetilde{\ket{1}\otimes\ket{0}}) -2 \ket{v_1}\right)$. 

\item For $\nu=2$,  $\widetilde{(\ket{1}\otimes\ket{1})}$
occurs as a summand as follows.
\begin{itemize}
\item For $k=2,~\nu'\geq 3$, in
$\frac{1}{\sqrt{2}}\left(\widetilde{\ket{1}\otimes\ket{1}} -
\widetilde{\ket{0}\otimes\ket{2}}\right)$ or  in
$\frac{1}{\sqrt{2}}\left(\widetilde{\ket{1}\otimes\ket{1}} -
\widetilde{\ket{2}\otimes\ket{0}}\right)$ according as $d_2=\nu'$
or $d_1=\nu'$. In fact, it is the same as
$\ket{a_{^2\mathbf{i}^0, ^2\mathbf{i}^1}}=
\frac{1}{\sqrt{2}} (\ket{^2\mathbf{i}^0}- \ket{
^2\mathbf{i}^1})$.
\item For $k\geq 3$, in
$\ket{a_{^2\mathbf{i}^0,^2\mathbf{i}^1}} =
\frac{1}{\sqrt{2}} ( \ket{^2\mathbf{i}^0} -
\ket{^2\mathbf{i}^1})$ and in $\ket{c_0^2} =
\frac{(|\mathcal{I}_2|-2)\left(\ket{^2\mathbf{i}^0} +
\ket{^2\mathbf{i}^1}\right)-2\ket{v_2}}{\sqrt{2(|\mathcal{I}_2|-2)|\mathcal{I}_2|}}$
\end{itemize}
\item For $\nu\geq 3, ~ \widetilde{\ket{1}\otimes\ket{1}}$
occurs as a summand in $\tilde{\ket{b_0^2}}$, and if, in
addition, $k\geq 3$, also in
$\ket{c_0^2} = \frac{|\mathcal{I}_2^2|\ket{\eta_2}
-|\mathcal{I}_2^1|\ket{v_2}}{\sqrt{|\mathcal{I}_2^1|
|\mathcal{I}_2^2| |\mathcal{I}_2|}}$.

\end{enumerate}
\end{enumerate}
\end{theorem}
\section{Entanglement properties of the projection operators} 
\label{projections}
We begin this section with some preparatory remarks,
which will be used to arrive at our main results.

\subsection{A useful involution on $\mathcal{I} \times
\mathcal{I}$.} \label{subsc:3.1}
\par Let $\mathcal{H}=\bigotimes_{t=1}^k \mathcal{H}_t$. Fix
$j$, with $1\leq j \leq k$. 

\par For $(\mathbf{p},\mathbf{q}) \in \mathcal{I} \times
\mathcal{I}$, let $\sigma_j (\mathbf{p},\mathbf{q}) =
(\mathbf{p'},\mathbf{q'})$, where 
\begin{equation*}
p_t' = \left\{\begin{array}{ll}
p_t & \text{for } t \neq j\\
q_j & \text{for } t = j
\end{array}\right.
\qquad\text{and}\qquad
q_t' = \left\{\begin{array}{ll}
q_t & \text{for } t \neq j\\
p_j & \text{for } t = j
\end{array}\right.
\end{equation*}
Then 
\begin{equation}\label{eq:pt}
|\mathbf{p}\rangle \langle \mathbf{q}|^{PT_j} =
|\mathbf{p'}\rangle \langle \mathbf{q'}|.
\end{equation}
We note that  $\sigma_j (\mathbf{q},\mathbf{p}) =
(\mathbf{q'},\mathbf{p'})$. Further, the map $\sigma_j \circ
\sigma_j$ is the identity map on $\mathcal{I} \times \mathcal{I}$,  i.e.,
the map $\sigma_j$ is an involution on $\mathcal{I} \times \mathcal{I}$.

\subsection{Action of $PT_j$. }\label{subsc:3.2}
\par Any operator $\rho \in \mathcal{L}(\mathcal{H}_1
\otimes \cdots \otimes \mathcal{H}_k)$ given as in 
(\ref{eq:pt1}) can be written in the compact form as,  
\begin{equation}\label{eq:state}
\rho = \sum_{ \mathbf{p},\mathbf{q} \in \mathcal{I}}
\rho_{(\mathbf{p},\mathbf{q})} |\mathbf{p}\rangle \langle
\mathbf{q}|,
\end{equation}
then
\begin{eqnarray*}
\rho^{PT_j} &=& \sum_{ (\mathbf{p},\mathbf{q}) \in \mathcal{I}
\times \mathcal{I}}
\rho_{(\mathbf{p},\mathbf{q})} |\mathbf{p'}\rangle \langle
\mathbf{q'}|\\
&=& \sum_{ (\mathbf{p},\mathbf{q}) \in \mathcal{I}
\times \mathcal{I}}
\rho_{\sigma_j (\mathbf{p},\mathbf{q})} |\mathbf{p}\rangle \langle
\mathbf{q}|.
\end{eqnarray*}

\par Fix $j'\neq j$ with $1 \leq j' \leq k$. Let
$\mathbf{p}^0 \in \mathcal{I}_0,~ \mathbf{q}^0 \in
\mathcal{I}_2$; $ \mathbf{p}^1$ and $\mathbf{q}^1 \in
\mathcal{I}_1$, be defined as

\[ p_t^0=0 \quad  \text{for all } t\qquad , \qquad 
q_t^0 = \left\{ \begin{array}{ll}
1 & \text{for  } t = j, j'\\
0 & \text{otherwise}
\end{array} \right.\] 
\[ p_t^1 = \left\{ \begin{array}{ll}
1 & \text{for  } t =j \\
0 & \text{otherwise}
\end{array} \right. \qquad , \qquad 
q_t^1 = \left\{ \begin{array}{ll}
1 & \text{for  } t= j'\\
0 & \text{otherwise}
\end{array} \right.\] 
Then $\sigma_j(\mathbf{p}^0, \mathbf{q}^0)
=(\mathbf{p}^1,\mathbf{q}^1)$. 
\par Let $\lambda
\neq 0$ be a real number. Set  $\ket{\xi}
=  \lambda \ket{\mathbf{p}^0} +\ket{\mathbf{q}^0}$. Then for any
$\mathbf{p,~q} \in \mathcal{I}$,
\begin{eqnarray*}
\langle \xi | \mathbf{p}\rangle \langle \mathbf{q} | \xi
\rangle &=&
(\lambda \delta_{\mathbf{p}^0\mathbf{p}}+\delta_{\mathbf{q}^0\mathbf{p}})
(\lambda \delta_{\mathbf{p}^0\mathbf{q}}+\delta_{\mathbf{q}^0\mathbf{q}})\\
&=& \left\{
\begin{array}{lll}
\lambda^2 & \text{for~} &(\mathbf{p},\mathbf{q})=
(\mathbf{p^0},\mathbf{p^0}), \\
\lambda  & \text{for~} &(\mathbf{p},\mathbf{q})\in
\{(\mathbf{p^0},\mathbf{q^0}),
(\mathbf{q^0},\mathbf{p^0})\}\\
1 & \text{for~} &(\mathbf{p},\mathbf{q})=
 (\mathbf{q^0},\mathbf{q^0})\\
0 && \text{otherwise} 
\end{array}\right.
\end{eqnarray*}
With a state $\rho$ as in (\ref{eq:state}), 
\begin{eqnarray}
\langle \xi|\rho^{PT_j}|\xi \rangle 
&=& \sum_{(\mathbf{p,q})\in \mathcal{I}\times \mathcal{I}}
\rho_{\sigma_j(\mathbf{p,q})} \langle \xi |
\mathbf{p}\rangle \langle \mathbf{q} | \xi \rangle \nonumber\\
&=& \lambda^2 \rho_{\sigma_j (\mathbf{p^0,p^0})}+ \lambda \rho_{\sigma_j
(\mathbf{p^0,q^0})} + \lambda \rho_{\sigma_j (\mathbf{q^0,p^0})} +
\rho_{\sigma_j (\mathbf{q^0,q^0})} \nonumber\\
&=& \lambda^2 \rho_{(\mathbf{p^0,p^0})}+ \lambda(\rho_{
(\mathbf{p^1,q^1})} + \rho_{(\mathbf{q^1,p^1})}) +
\rho_{(\mathbf{q^0,q^0})}.\nonumber
\end{eqnarray}

\begin{theorem}\label{th:main}
Let $\mathcal{H}=\bigotimes_{r=1}^k \mathcal{H}_r$.
 Let $P_\mathcal{S}$ be the projection  on
the completely entangled subspace $\mathcal{S}$. 
 For each $j$, $P_\mathcal{S}$ is not positive under
partial transpose at level $j$.

In particular, $P_\mathcal{S}$ is NPT.
\end{theorem}

\begin{proof}

\par For a unit vector $\ket{\zeta} \in \mathcal{H}$, let
$P_\zeta$ be the projection on $\ket{\zeta}$, i.e. $P_\zeta
=|\zeta\rangle \langle \zeta|$. Let $1 \leq j \leq k$. Take
any $j' \neq j$ with $1 \leq j' \leq k$.  Let $\mathcal{C}$ be an
orthonormal  basis for $\mathcal{S}$ 
in two separate cases as follows.
\begin{itemize}
\item[(a)] For $k=2,~ d_1 =d_2 = \nu$, take $\mathcal{C}=
\mathcal{B}$ as in \S\ref{ssuc:basiskrp}.
\item[(b)] For $k=2$ but $d_1 \neq d_2$, or $k\geq 3$ we follow
the procedure set up in \S \ref{ssu:or} for Theorem~\ref{th:3rd}.
 Then 
\[P_\mathcal{S} = \sum_{\ket{\zeta}\in \mathcal{C}} P_\zeta
=\sum_{\mathbf{p,q} \in \mathcal{I}} \rho_{(\mathbf{p,q})}
|\mathbf{p}\rangle\langle \mathbf{q}|,\]
for some suitable  $\rho_{(\mathbf{p,q})}$'s.
In the notation  \S\ref{subsc:3.2}, 
\begin{equation}\label{eq:9}
\langle \xi|\rho^{PT_j}|\xi \rangle =  \lambda^2 \rho_{(\mathbf{p^0,p^0})}+ \lambda(\rho_{
(\mathbf{p^1,q^1})} + \rho_{(\mathbf{q^1,p^1})}) +
\rho_{(\mathbf{q^0,q^0})}.
\end{equation}
To complete the proof it is enough to show that $\langle
\xi|\rho^{PT_j}|\xi \rangle<0$.

\par We arrange the elements of $\mathcal{C}$ in any manner
$\{\ket{\zeta_s}: 0\leq s  \leq M-1\}$, but insisting on the
following points.
\item[(c)] \[ \ket{\zeta_0} =\left\{
\begin{array}{lll}
\ket{a_{0,1}} && \text{ in case (a)}\\
&&\\
 \widetilde{\ket{a_{0,1}}} && \text{ in case (b)}.
\end{array}\right. \]
\item[(d)] \[ {\rm For~}  \nu = 2, \quad\quad
\ket{\zeta_1} = \ket{a_{^2\mathbf{i}^0, ^2\mathbf{i}^1}}
=\frac{1}{\sqrt{2}}(\ket{^2\mathbf{i}^0} -
\ket{^2\mathbf{i}^1}),\]
\[{\rm 
whereas~for~} \nu \geq 3, \quad \quad
\ket{\zeta_1} =\widetilde{\ket{b_0^2}}.\]
\item[(e)] \[{\rm For~} k\geq 3, \quad \quad
\ket{\zeta_2} = \ket{c_0^1}.\]
\item[(f)] \[{\rm For~} k \geq 3, \quad \quad
\ket{\zeta_3} = \ket{c_0^2}.\]
\end{itemize}
We write $P_s = P_{\ket{\zeta_s}}, ~ 0\leq s \leq M-1$. Then
$P_\mathcal{S}=\sum_{s=0}^{M-1}P_s$. So $\rho_{\mathbf{p,q}}
\neq 0$ only if $\ket{\mathbf{p}}$ and $\ket{\mathbf{q}}$ occur as a
summand in some $\ket{\zeta_s}$. 

\par In view of  Theorem~\ref{th:3rd}(iv) and
(\ref{eq:9}) above we can just confine our attention to the
vectors listed under (c), (d), (e) and (f) above.  

\par We first note that none of them contributes towards
$\rho_{\mathbf{(p^0,p^0)}}$. Also
$\rho_{\mathbf{(q^0,q^0)}}\geq 0$.  Next, we find that contribution
to $\rho_{(\mathbf{p^1,q^1})}$ is the same as that to
$\rho_{(\mathbf{q^1,p^1})}$. Thus, if the final contribution to
$\rho_{(\mathbf{p^1,q^1})}$ is $<0$, then for a suitable
$\lambda >0,~ \langle \xi| \rho^{PT_j} | \xi \rangle
<0$. We now proceed to show that it is so. 
 
\par $P_0$ contributes $-\frac{1}{2}$ to
$\rho_{\mathbf{p^1,q^1)}}$. For $k\geq
3$, $P_2$ contributes $\frac{1}{2}\frac{k-2}{k}$ to $\rho_{\mathbf{p^1,q^1)}}$.  So the total contribution to
$\rho_{(\mathbf{p^1,q^1})}$ is $-\frac{1}{k}$. Hence the
proof.
\end{proof}

\begin{corollary}
$\mathcal{F}$ does not contain any unextendable
orthonormal product basis. 
\end{corollary}
\begin{proof}
If $\mathcal{F}$ contains any unextandable product basis
then by Theorem~(\ref{th:bbd}) $P_{\mathcal{S}}$
will be PPT which is not true by Theorem~\ref{th:main}.  Hence
the result follows. 
\end{proof}

We now show that large classes of states with range in the completely
entangled subspace $\mathcal{S}$ are NPT.

\begin{theorem}\label{th:3.2}
Let $1 \leq j \leq k$. Take any $j' \neq j$ with $1 \leq j'
\leq k$. Any positive operator $\sum_{s=0}^{M-1} p_s P_s$, where
$ p_s \geq 0$ for all $s,~ p_0+(k-2)p_2 >0$ and $P_s$'s
are as in the proof of Theorem~\ref{th:main} above, is not
positive under partial transpose at level $j$.
\end{theorem}
\begin{proof} 
(i) All cases except possibly the case when
$k \geq 3$ and  $(k-2)p_2=k p_0$. 

\par We refer to the proof of Theorem~\ref{th:main} above. 
The only change needed is that
the term, say $w$ with $\lambda$ is now given as follows.

\begin{enumerate}[(a)]
\item For $k=2,~ w = - p_0$ (in place of $-1$),

\item In case $k\geq3,~ w= -p_0  +p_2 \frac{ k-2} {k } \neq
0$.
\end{enumerate}
So the final number in the right hand side of (\ref{eq:9})
can be made negative by suitable choice of $\lambda$ which
has to be suitably big and  $> 0$  if $w <  0$, and   has to be
$< 0$ and suitably big in absolute value  if $w >
0$.

(ii) Case $k \geq 3$ but $(k-2) p_2 = k
p_0$. Since $p_0 + (k-2)p_2 > 0$, we have $p_2 >0$. Because
$k\geq 3$, there is $j''$ with $j \neq j''
\neq j'$ and $1 \leq j'' \leq k$. Let $\mathbf{r}^0 \in
\mathcal{I}_2$ and $\mathbf{r}^1 \in  \mathcal{I}_1$ be given
by 
\[r_t^0 = \left\{\begin{array}{lll}
1 & \text { if } &t= j, ~ j'' \\
0 && \text{otherwise}
\end{array} \right. \] 
\[r_t^1 = \left\{ \begin{array}{lll}
1 & \text{ for } & t= j'' \\
0 && \text{ otherwise}.
\end{array}\right.\]
  We replace $\xi$ by $\xi'$ given by $\lambda \ket{
\mathbf{p}^0} + \ket { \mathbf{r}^0}$ with $\lambda$ real
and make computations similar to those in item 3.2 and proof of  part
(i) above. We note that $\mathbf{q}^1$ has to be replaced by
$\mathbf{r}^1$, and then $w$ by $w'= -\frac{2}{k} p_2$. And,
therefore, for $\lambda$ suitably bigger than
$0$, $\bra{ \xi'} \rho^{ PT_j} \ket{ \xi'} <0$.  This completes the proof.
\end{proof}
\begin{rem}
\noindent
\begin{enumerate}[(i)]
\item Because of the freedom of orthonormal bases at various
stages of the construction of $\mathcal{C}$ the import
of Theorem \ref{th:3.2} is much more. In fact, we may apply 
Theorem \ref{th:2nd} to construct a basis $\mathcal{D}$ for
$\mathcal{S}$ with more such freedom by clubbing in
$\mathcal{S}^{(n)}$'s,  $3\leq n \leq N-1$ and insisting on
including $\ket{\zeta_0},~\ket{\zeta_1}$, and in case $k\geq
3$, $\ket{\zeta_2}$ and $\ket{\zeta_3}$ as well. 

\item  Let  $1 \leq r \leq k$. Let  $\gamma_r$ be
the involution on the set $\mathcal{D}_r = \{p :0 \leq p
\leq d_r - 1\}$ to itself that takes $p \mapsto d_r -1 - p$ for
$0 \leq p \leq  d_r - 1$. This induces a unitary
linear operator $R_r$ on $\mathcal{H}_r$ to itself which takes $e_p$ to
$e_{\gamma_r (p)}$ for  $p\in \mathcal{D}_r$. We note that
$R_r^2=I_{\mathcal{H}_r}$ and therefore, $R_r$ is self-adjoint. 
Next, let  $\gamma= \prod_{r=1}^k \gamma_r$ on $
\mathcal{I}= \prod_{r = 1}^k \mathcal{D}_r$ to itself. Then
$\gamma$ is an involution on $\mathcal{I}$ to itself. Further,
for $0 \leq n \leq N$, $\gamma$ takes $\mathcal{I}_n$ to
 $\mathcal{I}_{N - n}$.  Let $R$ be the operator
$\bigotimes_{r = 1}^k R_r$ on $\mathcal{H}$ to itself. Then,
for $0 \leq n \leq N$, $R$
takes $\mathcal{H}^{(n)}$ onto  $\mathcal{H}^{(N - n)}$,
$u_n$ to $u_{N - n}$, $\mathcal{T}^{(n)}$ onto
$\mathcal{T}^{(N - n)}$, $\mathcal{S}^{(n)}$ onto
$\mathcal{S}^{(N-n)}$.
Therefore,  $R$ takes $\mathcal{S}$ onto itself.  Further, $R$
is  unitary and   self-adjoint. For $\mathbf{p}$, $\mathbf{q}
\in \mathcal{I},~ R (\ket{\mathbf{p}} \bra{ \mathbf{ q}}
) R= \ket{ \gamma( \mathbf{ p})} \bra{ \gamma(\mathbf{ q})}$.
Also, for $1 \leq j \neq j' \leq  k$, we may now
consider $R^\dag \rho R = R \rho R$ with $\rho$'s as
indicated  in Theorem \ref{th:3.2} and  part (i) above to add
to the class of positive operators with range in
$\mathcal{S}$ whose partial transpose at
level $j$ is not positive.

\item For $1 \leq j \leq k$ and $ 1 \leq j' \leq k$ with $j
\neq j'$ let $\mathcal{N}_{j,j'}$ be the set of NPT$_j$
states obtained
in Theorem \ref{th:3.2} together with those by methods
indicated in (i) and (ii) above. Put
\[\mathcal{N} =\bigcup_{\substack {1 \leq j \leq k \\ 1 \leq
j' \leq k \\ j \neq j'}} \mathcal{N}_{j,j'}.\]
Then each $\rho$ in $\mathcal{N}$ has range in the subspace
$\mathcal{S}$ and has a non-positive partial transpose at
some level. 

\item Johnston \cite{PhysRevA.87.064302} asked the following
question. 
\begin{quote}
What is the maximum dimension $\mu$ of a subspace
with the property that any state with range in the subspace
has at least one partial transpose which is non-positive.
\end{quote}
Let us call a subspace $\mathcal{E}$ of $\mathcal{H}$
satisfying this criteria an NPT space.

\item Let $\mathcal{E}$ be a subspace of $\mathcal{H}$. If $\{
\rho :~ \rho \text{  is a state with range in }
\mathcal{E}\}$ is contained in  $\mathcal{N}$, then $\mathcal{E}
\subset \mathcal{S}$ and $\mathcal{E}$ is NPT. In
particular, If $\mathcal{N} =\{\rho:~ \rho \text{ is a state
with range in } \mathcal{S}\}$, then $\mathcal{S}$ is NPT.
If that be so, then the answer to Johnston's question is 
\[\mu = M = d_1 d_2 \cdots d_k -(d_1 + d_2+ \cdots + d_k) +k -1.\]
This question still remains open, but the progress made in
this paper above does show that $\mathcal{N}$ is
substancially large.   
\end{enumerate}
\end{rem}

\section{Conclusion}
\label{conc}
Let $\mathcal{S}$ be a concrete completely entangled subspace of maximal
dimension, in $\mathcal{H} =\bigotimes_{j=1}^k \mathcal{H}_j$ with $2\leq d_j
=\dim \mathcal{H}_j < \infty$ for $1 \leq j \leq k$, constructed by
Parthasarathy~\cite{krp1}. Let $P_\mathcal{S}$ be the projection on this space.
We realized that the particular orthonormal basis $\mathcal{B}$ for
$\mathcal{S}$ for the bipartite case of equal dimensions obtained by
Parthasarathy~\cite{krp1} helps us to prove that $P_\mathcal{S}$ is not
positive under partial transpose. For any fixed $j$ and $j'$ with
$1\leq j \neq j' \leq k$,  we developed techniques to construct a
suitable orthonormal basis $\mathcal{C}$ for $\mathcal{S}$ for the multipartite
case utilizing $\mathcal{B}$ in the process. This enabled us to prove that
$P_\mathcal{S}$ is not positive under partial transpose at
level $j$. We next extended this to certain positive
operators $\rho$'s with range contained in
$\mathcal{S}$.  This generalizes a substantial  part of  the corresponding
result of Johnston~\cite{PhysRevA.87.064302} for the
bipartite case. Even after varying $j$ and $j'$ and clubbing
all $\rho$'s,  the
question whether there are any  states with support in
$\mathcal{S}$ that are PPT$_j$ for each $j$,  $ 1 \leq j
\leq k$, remains open.
However, in this paper we have made substantial progress in the direction
of obtaining an answer.
Further results on this issue will be presented elsewhere.
\bibliographystyle{alpha}
\bibliography{biblio}
\end{document}